\documentclass[journal]{IEEEtran}
\ifCLASSINFOpdf
\else
\fi

\usepackage{amsthm}
\usepackage{algorithm}
\usepackage{algorithmicx}
\usepackage{algpseudocode}
\usepackage{graphics,graphicx}
\usepackage{textcomp}
\usepackage{amsmath}
\usepackage{txfonts}
\usepackage{dsfont}
\usepackage{lipsum}

\algnewcommand\algorithmicinput{\textbf{Input:}}
\algnewcommand\Input{\item[\algorithmicinput]}

\algnewcommand\algorithmicoutput{\textbf{Output:}}
\algnewcommand\Output{\item[\algorithmicoutput]}
\newlength\myindent
\algnewcommand\algorithmicforeach{\textbf{for each}}
\algdef{S}[FOR]{ForEach}[1]{\algorithmicforeach\ #1\ \algorithmicdo}

\begin{document}
%
\title{Online Participatory Sensing in Double Auction Environment with Location Information}
%
%
%

\author{ Jaya~Mukhopadhyay, Vikash~Kumar~Singh,
        Sajal~Mukhopadhyay, Anita~Pal
\thanks{J. Mukhopadhyay is with the Department
of Mathematics, National Institute of Technology, Durgapur,
WB, 713209 India e-mail: (jayabesu@gmail.com).}
\thanks{V. K. Singh is with the Department
of Computer Science and Engineering, National Institute of Technology, Durgapur,
WB, 713209 India e-mail: (vikas.1688@gmail.com).}
\thanks{S. Mukhopadhyay is with the Department
of Computer Science and Engineering, National Institute of Technology, Durgapur,
WB, 713209 India e-mail: (sajmure@gmail.com).}
\thanks{A. Pal is with the Department
of Mathematics, National Institute of Technology, Durgapur,
WB, 713209 India e-mail: (anita.buie@gmail.com).}
}
\maketitle

\begin{abstract}
 As mobile devices have been ubiquitous, participatory sensing emerges as a powerful tool to solve many contemporary real life problems. Here, we contemplate the participatory sensing in online double auction environment by considering the location information of the participating agents. In this paper we propose a truthful mechanism in this setting and the mechanism also satisfies the other economic properties such as budget balance and individual rationality.
\end{abstract}


%
\IEEEpeerreviewmaketitle

\section{Introduction}
\IEEEPARstart{P}{articipatory sensing} \cite{e1}\cite{Lane:2010:SMP:1866991.1867010} is a distributed problem solving model in which the common people (may not be \emph{professionals}) of indefinite size carrying smart devices (such as Tablets, smart watches, smartphones, etc.) may be engaged to accomplish the tasks or sub-tasks. Examples include measuring the level of smokes and toxic gases present in the environment of certain industrial area \cite{Mun:2009:PPE:1555816.1555823}\cite{Massung:2013:UCS:2470654.2470708}, keeping track of auto-mobiles traffic condition in highly populated urban areas \cite{lan2013}\cite{Horvitz:2005:PES:3020336.3020371}, monitoring the state of the roads (eg. patholes, bumps, etc.) by attaching 
sensors to cars \cite{Smartphones_nericell:rich}, and several directions in healthcare and physical fitness \cite{6716722}\cite{conf:issnip:AflakiMBH13}. In a typical participatory sensing model, there exists three different participating community namely; (a) \emph{task requester(s)}, (b) \emph{platform (or third party)}, and (c) \emph{task executer(s)}. The working of participatory sensing system is initiated by the \emph{task requester(s)} who submit their sensory task(s) to the \emph{platform} that is/are to be accomplished by the common people incorporated with smart devices. Once the \emph{third party} (or \emph{platform}) receives the task(s), he/she (henceforth he) outsource the task(s) or subtask(s) to the group of common people carrying smart devices. In the participatory sensing terminology, these common people carrying smart devices are termed as \emph{task executer(s)}. In order to complete the assigned task(s), the \emph{task executers} have to utilize their owned smart device resources (such as \emph{battery}, \emph{GPS} system, etc.). Now, the obvious question that may arise is: how to motivate the \emph{task executer(s)} to accomplish the projected task(s) voluntarily by utilizing their resources. Moreover, it is to be noted that each of the \emph{task requesters} and the \emph{task executers} (synonymously called agents) are \emph{strategic}. In general, \emph{strategic} means that, the agents chooses their strategies so as to maximize a well defined individualistic utility.     
\begin{figure}[H]
\begin{center}
\includegraphics[scale=0.4]{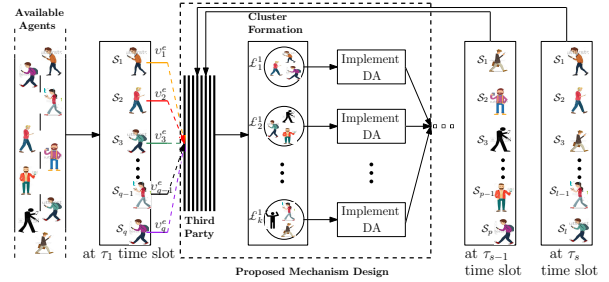}
\caption{System model}
\label{fig:model}
\end{center}
\end{figure}

\noindent  Answering to above posed question, for motivating the large group of \emph{task executers} for voluntary participation, one can think of the solution to incentivize the participating \emph{task executers} by some means, once the task(s) is/are completed. In this paper, we study a \textbf{s}ingle \textbf{t}ask \textbf{e}xecution \textbf{p}roblem (STEP); where there are multiple \emph{task requesters} having a single common task,  that is to be accomplished by the multiple \emph{task executers} in an \emph{online} environment. By \emph{online} environment, we mean that the agents arrives in the system and departs from the system on a regular basis. The proposed model is shown in \textbf{Fig.} \ref{fig:model}. The novelty that is introduced in this is to develop a game theoretic approach to model the STEP. As their are multiple \emph{task executers} and multiple \emph{task requesters}, this give rise to a double auction framework. In our model the location information of the agents are considered so as to cover a substantial area albeit collecting to much of redundant data. It is to be noted that the \emph{task executers} location information are tracked implicitly during the supply of the completed tasks to the \emph{third party}. By doing so it is guaranteed that the \emph{task executers} can't gain by lying their actual location.  The location aware participatory sensing was first introduced in \cite{l2}. However location aware participatory sensing in online double auction environment was not addressed in \cite{l2}. In this paper we have addressed the location aware participatory sensing in online DA. To avoid collecting redundant data, clustering concept is implemented before running the auction in each round.\\
\\
The main contributions of this paper are:
\begin{itemize}
\item In the participatory sensing scenario, we have proposed a framework to study the STEP with multiple \emph{task executers} and multiple \emph{task requesters} in an online environment by utilizing the concept of \emph{clustering} along-with auction. As their are multiple \emph{task requesters} and multiple \emph{task executers}, it is a good choice to model the participatory sensing scenario using double auction.    
\item We propose a \textbf{s}ingle \textbf{t}ask \textbf{e}xecution \textbf{m}echanism (STEM) for STEP motivated by \cite{DBLP:journals/corr/MukhopadhyayPMS16}\cite{Jbre_INte_2005} that takes into account multiple \emph{task executers} and multiple \emph{task requesters}. We design a \emph{truthful} (or \emph{incentive compatible}) mechanism for this interesting class
of problem.
\item We have shown that STEM is bounded above by $O(k  \kappa  n^2)$. Moreover, we have also shown that our STEM satisfies the several economic properties such as \emph{truthfulness}, \emph{individual rationality}, and \emph{Budget balance}. 
\item A substantial amount of simulation is done to compare STEM with the carefully designed benchmark scheme (McAfee's rule).
\item We have proved that in the given \emph{online} environment with clustering the agents can't gain by manipulating their valuation, arrival and departure time in a given arrival-departure window. 
\end{itemize} 
\noindent The remainder of the paper is structured as follows. Section II
elucidates the preliminary concepts about participatory sensing. Section III describes our proposed model. The proposed mechanisms is illustrated in section IV. The paper is concluded with the possible future directions in section V.

\section{Prior works}
 Recently, there has been a spate of research work at the border of participatory sensing and in their several applications areas. In this section we discuss the prior works on participatory sensing, taking into account \emph{incentives} aspects, quality of data or information supplied, privacy of the \emph{task executers} performing the task, and different set-up participatory sensing in budget constraint environment.\\
 In order to get a nice overview of the participatory sensing the readers may refer \cite{Restuccia:2016:IMP:2925994.2888398}\cite{e1}\cite{Lane:2010:SMP:1866991.1867010}\cite{journals:cm:GantiYL11}\cite{10.1109:MIC.2012.70}. Currently, the participatory sensing is one of the open research areas. One obvious question that arise in the participatory sensing environment is: how to motivate the large common people carrying smart devices to participate in the system? To answer this question in a better way, the researchers have provided their immense effort in this direction. In past, for voluntary participation of the \emph{task executers} several incetivizing schemes are discussed in literature.
\cite{Reddy:2010:EMP:1864349.1864355} follows the fixed price payment scheme, where the winning agents are paid a fixed price as their payment. However, the fixed price based incentive scheme may not satisfy the several participating agents because of the amount of effort they make in the data collection process. Moreover, the incentive based schemes has got a special attention from the research community. \cite{TLHP_INFO_2014} addresses the incentive
scheme under the reverse auction based setting (single buyer
and multiple sellers). 	
Several incentive schemes has been introduced in \cite{DBLP:journals/corr/abs-1305-6705}\cite{Lee2010PMC} \cite{z1}. In \cite{jas1}\cite{r1}\cite{v1} efforts have been made by the researchers to show the effect of quality of data collected by the agents to the overall system by incorporating the quality of data to the system in some sense. Some initial research has been carried out by \cite{e1} \cite{v1} \cite{par1} \cite{par2} to preserve the privacy of the agents so that their private information associated with the data are not revealed publicly. Recently, \cite{l2} provides the incentive schemes under the location constraints. in their work they have addressed location aware participatory sensing in one buyer and multiple seller environment. In our model we have explored more general multiple sellers-multiple buyers framework in more challenging location aware participatory sensing in online double auction environment.

\section{System model}
In this section, considering an online environment we formalize a \textbf{s}ingle \textbf{t}ask \textbf{e}xecution \textbf{p}roblem (STEP) for the participatory sensing scenario. By online environment, we mean that the agents arrives in the auction market and departs from the auction market on a regular basis in a given time horizon $\varmathbb{T}$ (say \emph{a day}). Let $\boldsymbol{\mathcal{B}} = \{\boldsymbol{\mathcal{B}}_1, \boldsymbol{\mathcal{B}}_2, \ldots, \boldsymbol{\mathcal{B}}_m \}$ be the set of \emph{task requesters} and $\mathcal{S} = \{\mathcal{S}_1, \mathcal{S}_2, \ldots, \mathcal{S}_n\}$ be the set of \emph{task executers} such that $m << n$. The \emph{task executers} and \emph{task requesters} are synonymously called agents. Each of the \emph{task executer} $\mathcal{S}_i$ incurs a private cost for performing the available task termed as \emph{valuation} given as $\upsilon_i^{e}$. The set $\upsilon^e$ denotes the set of valuations of all the \emph{task executers} given as $\upsilon^e = \{\upsilon_1^e, \upsilon_2^e, \ldots, \upsilon_n^e\}$. Similar to the \emph{task executers}, each of the \emph{task requester} $\boldsymbol{\mathcal{B}}_i$ has some private value for buying the task after its completion and is given as $\upsilon_{i}^{r}$. The set $\upsilon^r$ denotes the set of valuations of all \emph{task requesters} given as $\upsilon^r = \{\upsilon_1^r, \upsilon_2^r, \ldots, \upsilon_m^r\}$. Each of the \emph{task executers} and \emph{task requesters} places their private information in a sealed bid manner. It is to be noted that, due to the strategic nature of the agents, they can mis-report their respective private values. So, it is convenient to represent the cost reported for performing the task by the \emph{task executer} $\mathcal{S}_i$ as $\hat{\upsilon}_i^{e}$ and the value of \emph{task requester} $\mathcal{B}_i$ for buying the task as $\hat{\upsilon}_{i}^{r}$. $\hat{\upsilon}_{i}^{e} = \upsilon_{i}^{e}$ and $\hat{\upsilon}_{i}^{r} = \upsilon_{i}^{r}$ describes the fact that $\mathcal{S}_i$ and $\mathcal{B}_i$ are not deviating from their true valuations. In this model, there are multiple \emph{task requesters} (as \emph{buyers}) and multiple \emph{task executers} (as \emph{sellers}). So, this is a perfect setting to model the STEP as an online double auction problem (ODAP). Due to online nature of the STEP, one of the realistic parameters that is perceived in our proposed model is arrival and departure time of the agents. The arrival time of any agent is the time at which he/she (henceforth he) knows about the auction market or the time at which he first become aware of his desire to involve into the auction market after entering into the system. The arrival time of each \emph{task executer} $\mathcal{S}_i$ and each \emph{task requester} $\boldsymbol{\mathcal{B}}_i$ are given as $a_{i}^{e}$ and $a_{i}^{r}$ respectively. For a \emph{task requester}, we interpret the departure time as the final time in which he values the task. For a \emph{task executer}, the departure time is the final time in which he is willing to accept payment. The departure time of each \emph{task executer} $\mathcal{S}_i$ and each \emph{task requester} $\boldsymbol{\mathcal{B}}_i $ are given as $d_{i}^{e}$ and $d_{i}^{r}$ respectively. The agents may mis-report their respective arrival time or departure time or both within the arrival-departure window in order to gain. It is convenient to represent the arrival time of \emph{task executer} $\mathcal{S}_i$ and \emph{task requester} $\boldsymbol{\mathcal{B}}_i$ as $\hat{a}_{i}^{e}$ and $\hat{a}_i^{r}$ respectively. Similarly, more conveniently the departure time of \emph{task executer} $\mathcal{S}_i$ and \emph{task requester} $\boldsymbol{\mathcal{B}}_i$ is $\hat{d}_{i}^{e}$ and $\hat{d}_i^{r}$ respectively. $\hat{a}_{i}^{e} = a_{i}^{e}$, $\hat{a}_{i}^{r} = a_{i}^{r}$, $\hat{d}_{i}^{e} = d_{i}^{e}$, and $\hat{d}_{i}^{r} = d_{i}^{r}$ describes the fact that $\mathcal{S}_i$ and $\boldsymbol{\mathcal{B}}_i$ are not misreporting their arrival and departure time.
In our proposed mode\newtheorem{lemma}{Lemma}l, a day is termed as time horizon $\varmathbb{T}$. The time horizon $\varmathbb{T}$ is partitioned into several time slots (not necessarily of same length) given as $\varmathbb{T} = \{\tau_1, \tau_2, \ldots, \tau_s\}$. For each time slot $\tau_i$, a new set of active \emph{task requesters} $\mathcal{R} \subset \boldsymbol{\mathcal{B}}$ and a new set of active \emph{task executers} $\mathcal{U} \subset \mathcal{S}$ arrives in the auction market. At each time slot $\tau_i$, considering the newly active \emph{task executers} $\mathcal{U}$, a set of clusters of \emph{task executers} are formed and is given as: $\mathsterling^i = \{\mathsterling_1^i, \mathsterling_2^i, \ldots, \mathsterling_k^i\}$; where $\mathsterling_j^i$ is termed as the $j^{th}$ cluster for $\tau_i$ time slot. Over the $\varmathbb{T}$ time horizon the cluster vector can be given as: $\mathsterling = \{\mathsterling^{1}, \mathsterling^{2}, \ldots, \mathsterling^{s}\}$. Once the clusters are formed, then for each cluster $\mathsterling_j^i$ several independent double auction will be performed. At each time slot $\tau_i \in \varmathbb{T}$ and from each cluster $\mathsterling_j$ the set of winning \emph{task executers}-\emph{task requesters} are paired. At each time slot $\tau_i \in \varmathbb{T}$, 
our proposed mechanism matches one \emph{task executer} to one \emph{task requester} in a cluster.
More formally, a mechanism $\mathcal{M}$ = ($\mathcal{A}$, $\mathcal{P}$), where, $\mathcal{A}$ is called an allocation function and $\mathcal{P}$ is called a payment function. The allocation function $\mathcal{A}$ maps the pair of \emph{task executers} valuation and \emph{task requesters} valuation to the possible \emph{task executer-task requester} pairs. Following the payment function, the payment of each \emph{task executer} $\mathcal{S}_i$ and each \emph{task requester} $\boldsymbol{\mathcal{B}}_i$ is given as $\mathcal{P}_i^{e}$ and $\mathcal{P}_i^{r}$ respectively. As the \emph{task executers} and \emph{task requesters} are strategic in nature, they will try to maximize their utility. The utility of any \emph{task executer} is defined as the difference between the p\let\oldReturn\Return
\renewcommand{\Return}{\State\oldReturn}ayment received by the \emph{task executer} and the true valuation of the \emph{task executer}. More formally, the utility of $\mathcal{S}_i$ is  $\varphi^e_i = \mathcal{P}_i^{e}-\upsilon_i^{e}$, if $\mathcal{S}_i$ wins otherwise 0.
Similarly, the utility of any \emph{task requester} is defined as the difference between the true valuation of the \emph{task requester} and the payment he pays. More formally, the utility of $\boldsymbol{\mathcal{B}}_i$ is $\varphi^r_i = \upsilon_i^{r} - \mathcal{P}_i^{r}$, if $\boldsymbol{\mathcal{B}}_i$ wins 0 otherwise.

\section{STEM: Proposed mechanism}
\subsection{Outline of STEM}
In order to present the brief idea of the STEM to the readers the outline of the STEM is discussed before going into the detailed view. The outline of the STEM can be thought of as a three stage process:
\begin{itemize}
\item [$\blacksquare$]For any auction round $t \in \varmathbb{T}$ find out the active \emph{task executers} and \emph{task requesters}.
\item [$\blacksquare$]Cluster the active \emph{task executers} based on \emph{k-means} clustering technique.
\item [$\blacksquare$]Run the online double auction separately for each cluster of \emph{task executers}. \emph{Task requesters} will be the same for all the clusters.  
\end{itemize}
\subsection{Sketch of the STEM}
The three stage STEM can further be studied under four different sections: \emph{Main routine}, \emph{Cluster formation}, \emph{Payment}, and \emph{Allocation}. First, the sub-part of the proposed mechanism $i.e.$ the \emph{Main routine} phase is discussed and presented. The \emph{Cluster formation}
phase is addressed next. Next, the crucial part of the proposed mechanism $i.e.$
\emph{payment} phase motivated by \cite{Jbre_INte_2005} is discussed and presented. Finally, one of the \emph{allocation}
phase is addressed.
\subsubsection{Main routine}
The idea lies behind the construction of \emph{Main routine} is to handle the system partitioned into different time slots $\tau_i \in \varmathbb{T}$. The input to the \emph{Main routine} are the set of \emph{task executers} at $\tau_i$ time slot $i.e.$ $\mathcal{S}_{\tau_i}$, the set of available \emph{task requesters} at $\tau_i$ time slot $i.e.$ $\mathcal{B}_{\tau_i}$, the overall time horizon $i.e.$ $\varmathbb{T}$, the set of cost of execution of all \emph{task executers} $i.e.$ $\hat{\upsilon}^{e}$, and the set of value for buying the executed tasks by all the \emph{task requesters} $i.e.$ $\hat{\upsilon}^{r}$. The output is the set of allocation vector $\mathcal{A}$. In line 2, the several data structures that are utilized in \emph{main routine} are set to $\phi$. The \emph{for} loop in line 3 iterates over all the time slots $\tau_i \in \varmathbb{T}$. In line 4, the $active\_TE()$ function returns the set of active \emph{task executers} at time slot $\tau_i$ and is held in $\mathcal{U}$ data structure. Whereas, the set of active \emph{task requesters} at any time slot $\tau_i$ is determined by the function $active\_TR()$ and is held in $\mathcal{R}$ data structure. The \emph{for} loop in line $6-8$ iterates over the set of active \emph{task executers} $\mathcal{U}$ and keeps track of costs of the members in set $\mathcal{U}$ in $\gamma^e$ data structure. Similarly, the \emph{for} loop in line $9-11$ iterates over the set of active \emph{task requesters} $\mathcal{R}$ and keeps track of values of the members in set $\mathcal{R}$ in $\gamma^r$ data structure.

\begin{algorithm}[H]
\caption{Main routine($\mathcal{S}_{\tau_i}$, $\boldsymbol{\mathcal{B}}_{\tau_i}$, $\varmathbb{T}$, $\hat{\upsilon}^e$, $\hat{\upsilon}^r$)}
\begin{algorithmic}[1]
   \Output $\mathcal{A}$ $\leftarrow$ $\{\mathcal{A}_{1}, \mathcal{A}_2, \dots, \mathcal{A}_k\}$
\State \textbf{begin}
    \State $\mathcal{U} \leftarrow \phi$, $\mathcal{R} \leftarrow \phi$, $\gamma^e \leftarrow \phi$, $\gamma^r \leftarrow \phi$, $\mathsterling_{e}^* \leftarrow \phi$, $\mathsterling_{r}^* \leftarrow \phi$ 
    \ForAll{$\tau_i \in \varmathbb{T}$}
        \State $\mathcal{U} \leftarrow active\_TE~(\mathcal{S}_{\tau_i}, \tau_i)$  \Comment $\forall \mathcal{S}_i \in \mathcal{U}$, $\hat{a}_i^{e} \leq \tau_i < \hat{d}_i^{e}$
        \State $\mathcal{R} \leftarrow active\_TR~(\mathcal{B}_{\tau_i}, \tau_i)$ \Comment $\forall \mathcal{B}_i \in \mathcal{R}$, $\hat{a}_i^{r} \leq \tau_i < \hat{d}_i^r$
        \For{each $\mathcal{S}_i \in \mathcal{U}$}
            \State $\gamma^e \leftarrow \gamma^e \cup \mathcal{S}_i \cdot \hat{\upsilon}_{i}^{e}$ \Comment{$\hat{\upsilon}_{i}^{e}$  is the valuation \hspace*{50mm} component of $\mathcal{S}_i$ } 
        \EndFor
        \For{each $\boldsymbol{\mathcal{B}}_i \in \mathcal{R}$}
            \State $\gamma^r \leftarrow \gamma^r \cup \boldsymbol{\mathcal{B}}_i \cdot \hat{\upsilon}_{i}^{r}$ \Comment{$\hat{\upsilon}_{i}^{r}$ is the valuation \hspace*{50mm} component of $\boldsymbol{\mathcal{B}_i}$ }
        \EndFor
        \State $\mathsterling^i$ $\leftarrow$ Cluster formation ($\mathcal{U}$, $k$)
        \For {each $\mathsterling_j^i \in \mathsterling^i$}
        \State $\mathcal{U}_c \leftarrow Sort\_ascend$($\mathsterling_{j}^i$, $\mathcal{S}_i \cdot \gamma_i^e$) \Comment Sorting based on $\gamma_i^{e} \in \gamma^e$ for all $\mathcal{S}_i \in \mathsterling_{j}^i$ 
        \State $\mathcal{R} \leftarrow Sort\_descend$($\mathcal{R}$, $\mathcal{B}_i \cdot \gamma_i^r$) \Comment Sorting based on $\gamma_i^{r} \in \gamma^r$ for all $\mathcal{B}_i \in \mathcal{R}$ 
        \State Payment ($\mathcal{U}_c, \mathcal{R}$)
        \State $\mathcal{U}_c'^{(j)}$ $\leftarrow$ $\mathcal{U}_c'^{(j)}$ $\cup$ $\mathcal{U}_c^{*}$
        \State $\mathcal{R}'^{(j)}$ $\leftarrow$ $\mathcal{R}'^{(j)}$ $\cup$ $\mathcal{R}_c$
        \State $\mathsterling_e^{*}$ $\leftarrow$ $\mathsterling_e^{*}$ $\cup$ $\mathcal{U}_c'^{(j)}$
        \State $\mathsterling_r^{*}$ $\leftarrow$ $\mathsterling_r^{*}$ $\cup$ $\mathcal{R}'^{(j)}$
        \State $\mathcal{U}_c \leftarrow \phi$
        \EndFor   
        \State $\gamma^e \leftarrow \phi$, $\gamma^r \leftarrow \phi$
        \State New task executers and task requesters comes.
        \State $\mathcal{S}_{\tau_i}$ $\leftarrow$ $\mathsterling_{e}^{*} \cup \{new~task~executers\}$
        \State $\mathcal{B}_{\tau_i}$ $\leftarrow$ $\mathsterling_{r}^{*} \cup \{new~task~requesters\}$
    \EndFor        
    \Return $\mathcal{A}$
\State \textbf{end}
\end{algorithmic}
\end{algorithm}
 In line 12, a call to Cluster formation($\mathcal{U}$, $k$) is made; where \emph{k} is the number of clusters to be formed. Once in a particular time slot $\tau_i$ the cluster set $\mathsterling^i$ is formed in line 12, the payment and based on the payment the allocation is determined which is captured by the \emph{for} loop in line 13-22. In line 14 and 15 the \emph{task executers} and \emph{task requesters} are sorted in ascending and descending order respectively.
 In line 16, a call to \emph{payment} phase is done. In line 17 and 18 all the active task executers and task requesters at time $\tau_i$ in $j^{th}$ cluster which are not paired are placed in $\mathcal{U}_{c}'^{(j)}$ and in $\mathcal{R}'^{(j)}$ respectively.
 The $\mathsterling_{e}^*$ and $\mathsterling_{r}^*$ data structures keeps track of all the active \emph{task executers} and \emph{task requesters} in a given time slot $\tau_i$ but not allocated in there respective clusters. In line 21, the $\mathcal{U}_c$ data structure is set to $\phi$.
 The data structure $\gamma^e$ and $\gamma^r$ are set to $\phi$. Now, the new \emph{task executers} and new \emph{task requesters} are arriving in the market for the next time slot as depicted in line 24. In line 25-26 $\mathcal{S}_{\tau_i}$ and $\mathcal{B}_{\tau_i}$ captures the set of all \emph{task executers} and \emph{task requesters} that are going to participate in the next time slot. Line 28 returns the final allocation set $\mathcal{A}$.
\subsubsection{Cluster formation}
The input to the \emph{Cluster formation} are the set of active task executers at any time slot $\tau_i$ given as $\mathcal{U}$, and the number of cluster to be formed $i.e$ \emph{k}. Considering the \emph{centroid determination} phase, Line 2 initializes the $\mathcal{C}$ data structure utilized in the \emph{Cluster formation} algorithm. The \emph{random()} function in line 4 randomly picks a point as a centroid from the available point set $\mathcal{X}$. The randomly selected centroid is placed in $\mathcal{C}$ data structure using line 5. 
\begin{algorithm}[H]
\caption{Cluster formation ($\mathcal{U}$, $k$)}
\begin{algorithmic}[1]
   \State \textbf{begin}
	  
	\State $\mathcal{C} \leftarrow \phi$ \Comment \textbf{k centroid determination}
	\While{$|\mathcal{C}| \neq k$}
	\State $x^* \leftarrow random(\mathcal{X})$ \Comment Picking a random point $\mathcal{X}_{\ell} \in \mathcal{X}$
	\State $\mathcal{C} \leftarrow \mathcal{C} \cup \{x^*\}$
	\EndWhile
	\Repeat  \Comment \textbf{k cluster formation}
	\State $\mathsterling^i \leftarrow \phi$, $\mathsterling_j^i \leftarrow \phi$
	\For{each $\mathcal{S}_k \in \mathcal{U}$}
	  \For{each $\mathcal{X}_j \in \mathcal{C}$}
	  \State $D' \leftarrow D' \cup \{D(\mathcal{S}_k, \mathcal{X}_j)\}$ \Comment{\textbf{Distance between \hspace*{55mm} $\mathcal{S}_k$ and $\mathcal{X}_j$}}
	\EndFor	  
	\State $j^* \leftarrow argmin_j$ $D'$
	  \State $\mathsterling_{j^*}^i \leftarrow \mathsterling_{j^*}^i \cup \{\mathcal{S}_k\}$
	\EndFor
	\State $\mathcal{C} \leftarrow \phi$
	\For{$j=1$ to $k$}
	\State $\mathsterling^i \leftarrow \mathsterling^i \cup \mathsterling_j^i$
	\EndFor
  \For{each $\mathsterling_j^i \in \mathsterling^i$}
    \State $\mathcal{X}'_j = \frac{1}{|\mathsterling_{j}^{i}|} \sum_{x_\ell \in \mathsterling_j^i} x_\ell$ \Comment{$x_\ell$ is the point $\ell$ i.e. a two \hspace*{35mm}dimensional vector in cluster $\mathsterling_{j}^{i}$}
    \State $\mathcal{C} \leftarrow \mathcal{C} \cup \mathcal{X}'_j$
  \EndFor
  \Until{change in cluster takes place}
  \Return $\mathsterling^i$
\State \textbf{end}
\end{algorithmic}
\end{algorithm}
 In line 3 the \emph{while} loop ensures that the loop terminates when the size of $k$-centroid are determined. Line 8 initializes the $\mathsterling_{j}^i \leftarrow \phi$ and $\mathsterling^i \leftarrow \phi$. 
The \emph{for} loop in line 9 iterates over the active set of \emph{task executers} $\mathcal{U}$. Line 10-12 determines the closest distance of each $\mathcal{S}_k$ to a centroid $\mathcal{X}_j$. The $\mathsterling_{j^*}^{i}$ in line 14 keeps track of each $\mathcal{S}_k$. In line 16, the $\mathcal{C}$ data structure is set to $\phi$. The $\mathsterling^i$ data structure keeps track of all the clusters formed in a particular time slot $\tau_i$ as depicted in line $17-19$. Line $20-23$ determines the new centroids. The procedure in line 7-24 will be repeated until the change in the cluster is seen. Line 25 returns the set of cluster in any particular time slot $\tau_i$.  

\subsubsection{Payment}
The input to the \emph{payment} phase are the set of winning \emph{task executers} $i.e.$ $\mathcal{U}_c$, and the set of winning \emph{task requesters} $i.e.$ $\mathcal{R}$. In line 2 the $\hat{\mathcal{U}}$ and $\hat{\mathcal{R}}$ data structure are set to $\phi$. The \emph{for} loop in line 3-15 keeps track of payment of the winning \emph{task executers}. The check in line 3 confirms that if the \emph{task executer} belongs to freshly arrived category then the payment is decided by line 5 of the Algorithm 4 otherwise the payment is made using line 7.  Now, the check in line 9 is done to guarantee that the payment made to any task executer $\mathcal{S}_i$ is greater than its cost for executing the task $i.e.$ satisfying the important economic property \emph{individual rationality}. 
\begin{algorithm}[H] 
\caption{Payment ($\mathcal{U}_c$, $\mathcal{R}$)}
\begin{algorithmic}[1]
\State \textbf{begin}
\State $\hat{\mathcal{U}} \leftarrow \phi$, $\hat{\mathcal{R}} \leftarrow \phi$
        \For{each $\mathcal{S}_i \in \mathcal{U}_c$}
            \If{$\hat{a}_i^{e} == \tau_i$} \Comment{\textbf{Fresh arrival}}
               \State $\chi_{i}^{e} \leftarrow \min_{\rho^e \in [\hat{d}_i^e-\kappa, ~\tau_i]} \{\mathcal{P}_{i}^{e}(\rho^e)\}$
            \Else       \Comment{\textbf{Still active}}
                \State $\chi_{i}^{e}$ $\leftarrow$ $\min \{\mathcal{P}_{i}^{e}(\tau_i-1), \mathcal{P}_{i}^{e}(\tau_i)\}$
            \EndIf   
            \If{$\chi_{i}^{e} \geq \hat{\upsilon}_{i}^{e}$}
               \State $\mathcal{P}_e$ $\leftarrow$ $\mathcal{P}_e \cup \{\chi_{i}^{e}\}$
               \State $\hat{\mathcal{U}} \leftarrow \hat{\mathcal{U}} \cup \{\mathcal{S}_i\}$
            \Else:
               \State $\mathcal{S}_i$ is priced out.
            \EndIf         
        \EndFor 
        \For{each $\mathcal{B}_i \in \mathcal{R}$}
            \If{$\hat{a}_i^{r} == \tau_i$} \Comment{\textbf{Fresh arrival}}
               \State $\chi_{i}^{r} \leftarrow \max_{\rho^r \in [\hat{d}_i^r-\kappa, ~\tau_i]} \{\mathcal{P}_{i}^{r}(\rho^r)\}$
            \Else   \Comment{\textbf{Still active}}
                \State $\chi_{i}^{r}$ $\leftarrow$ $\max \{\mathcal{P}_{i}^{r}(\tau_i-1), \mathcal{P}_{i}^{r}(\tau_i)\}$
            \EndIf   
            \If{$\chi_{i}^{r} \leq \hat{\upsilon}_{i}^{r}$}
               \State $\mathcal{P}_r$ $\leftarrow$ $\mathcal{P}_r \cup \{\chi_{i}^{r}\}$
               \State $\hat{\mathcal{R}} \leftarrow \hat{\mathcal{R}} \cup \{\mathcal{B}_i\}$
            \Else:
               \State $\mathcal{B}_i$ is priced out.
            \EndIf         
        \EndFor
        \State Allocation($\hat{\mathcal{U}}$, $\hat{\mathcal{R}}$, $\gamma^{e}$, $\gamma^r$)
\State \textbf{end}
\end{algorithmic}
\end{algorithm} 
The $\mathcal{P}_e$ data structure in line 10 keeps track of all the payment of all the winning \emph{task executers} satisfying the \emph{individual rationality}. If the condition in line 9 is not satisfied by the \emph{task executers}, then the winning \emph{task executers} are priced out of the market as depicted in line 13. The \emph{for} loop in line 16-28 keeps track of payment of the winning \emph{task requesters}. The check in line 17 confirms that if the \emph{task requester} belongs to freshly arrived category then the payment is decided by line 17 otherwise the payment is made using line 19. Now, the check in line 21 is done to guarantee that the payment made by any task executer $\mathcal{S}_i$ is no more than its value for buying the completed task $i.e.$ satisfying the important economic property \emph{individual rationality}. The $\mathcal{P}_r$ data structure in line 22 keeps track of all the payment of all the winning \emph{task requesters} satisfying the \emph{individual rationality}. If the condition in line 22 is not satisfied then the winning \emph{task requester} is priced out of the market as depicted in line 26. Finally, a call to the \emph{allocation phase} is done line 29.      
\paragraph*{Payment function}
For determining the payment of each agent the valuation of the first losing \emph{task executer} and losing \emph{task requester} is taken into consideration which is given by $ \mathcal{I}_j^{*} = argmax_{i} \{\gamma_{i}^{r} - \gamma_i^{e} < 0\}$ such that $\gamma_{i}^{r} = \boldsymbol{\mathcal{B}}_i \cdot \hat{\upsilon}_{i}^{r}$ and $\gamma_{i}^{e} = \mathcal{S}_i \cdot \hat{\upsilon}_{i}^{e}$.
For defining the payment, we further require to 
fetch the valuation of the \emph{task requester} and the \emph{task executer} at the index position $\mathcal{I}_{j}^{*}$. The valuation of the \emph{task requester} at any index position is captured by the bijective function $\Upsilon^r: \mathbb{Z} \rightarrow \mathbb{R}_{\geq 0}$, whereas the valuation of the \emph{task executer} at any index position is captured by the bijective function $\Upsilon^e: \mathbb{Z} \rightarrow \mathbb{R}_{\geq 0}$. 
Let us further denote the valuation of the \emph{task requester} at the index position $\mathcal{I}_{j}^{*}$ 
by $\Upsilon^{r}(\mathcal{I}_{j}^{*})$ and the valuation of the task executer at $\mathcal{I}_{j}^{*}$ by $\Upsilon^{e}(\mathcal{I}_{j}^{*})$. For determining the payment of all winning \emph{task executers} and \emph{task requesters} we will take the help of the average of the cost of the task executer at $\mathcal{I}_{j}^*$ and the value of the \emph{task requester} at $\mathcal{I}_{j}^*$ given as $\eta$=$\frac{\Upsilon^{r}(\mathcal{I}_{j}^{*})+\Upsilon^{e}(\mathcal{I}_{j}^{*})}{2}$.
Mathematically, the payment of $i^{th}$ \emph{task executer} is given as:
\begin{equation}                   
 \mathcal{P}_{i}^{e}(\tau_i) =
  \begin{cases}
   \eta,        & \text{if $\Upsilon^e(\mathcal{I}_{j}^*) \geq \eta$ and $\Upsilon^r(\mathcal{I}_{j}^*) \leq \eta$ }\\
   \Upsilon^e(\mathcal{I}_{j}^*), & \text{otherwise}
  \end{cases}
  \end{equation}
Similarly, the payment of the $i^{th}$ \emph{task requester} is given as:
\begin{equation}                   
 \mathcal{P}_{i}^{r}(\tau_i)=
  \begin{cases}
   \eta,        & \text{if $\Upsilon^e(\mathcal{I}_{j}^*) \leq \eta$ and $\Upsilon^r(\mathcal{I}_{j}^*) \geq \eta$ }\\
   \Upsilon^r(\mathcal{I}_{j}^*), & \text{otherwise}
  \end{cases}
  \end{equation}
In this problem set-up, for any particular time slot $\tau_i \in \varmathbb{T}$ there might be two types of agents: (a) \emph{Freshly arrived} agents, 
(b) \emph{Still active} agents.  
For \emph{freshly} arrived \emph{task executers} and \emph{task requesters} the payment is calculated as shown below. More formally, the payment of $i^{th}$ \emph{task requester} is given as: 
\begin{equation}\label{pr1}                   
 \zeta^r(\tau_i) =
  \begin{cases}
  \max_{\rho^r \in [\hat{d}_i^r-\kappa,\ldots,~\tau_i] }  \{\mathcal{P}_{i}^{r}(\rho^r)\}, & \text{if task requester is freshly} \\ &\text{ arrived }\\
   \max \{\zeta^r(\tau_{i-1}), \mathcal{P}_{i}^{r}(\tau_i)\}, & \text{if task requester are still} \\ &\text{ active}
  \end{cases}
  \end{equation}
Here $\kappa$ is the maximum permitted gap between the arrival and departure of any arbitrary agent \emph{i}.
\begin{equation}\label{pr1}                   
 \zeta^e(\tau_i) =
  \begin{cases}
  \min_{\rho^e \in [\hat{d}_i^e-\kappa,\ldots,~\tau_i]} \{\mathcal{P}_{i}^{e}(\rho^e)\}, & \text{if task executer is freshly} \\ &\text{arrived }\\
   \min \{\zeta^e(\tau_{i-1}), \mathcal{P}_{i}^{e}(\tau_i)\}, & \text{if task executers are still} \\ &\text{ active}
  \end{cases}
  \end{equation}

Now, if after $\tau_i$ time slots if a task requester $i$ is a winner then the final payment of that task requester will be given by   $\mathcal{P}_{i}^r(\tau_i)=\max \{\zeta^r(\tau_{i-1}), \mathcal{P}_{i}^{r}(\tau_i)\}$
and similarly if after $\tau_i$ time slots if a task executer $i$ is a winner then the final payment of that task executer will be given by  $\mathcal{P}_{i}^e(\tau_i)=\min \{\zeta^e(\tau_{i-1}), \mathcal{P}_{i}^{e}(\tau_i)\}$.
\subsubsection{Allocation}
The input to the \emph{allocation phase} are the $j^{th}$ cluster in $\tau_i$ time slot $i.e.$ $\mathsterling_{j}^i$, the set of \emph{task requester} $\mathcal{R}$, the set of cost of task execution of \emph{task executers} $i.e.$ $\gamma^{e}$, and the set of values of \emph{task requesters} $i.e.$ $\gamma^r$. The output is the set of \emph{task requester}-\emph{task executer} winning pairs held in $\mathcal{A}_k$. Line 3 sorts the cluster $\mathsterling_{j}^i$ in ascending order based on the elements of $\mathcal{P}_e$ and held in $\mathcal{U}_c^{*}$ data structure. The set of active \emph{task requesters} are sorted in descending order based on the elements of $\mathcal{P}_r$ and held in $\mathcal{R}_c$ data structure. 
\begin{algorithm}[H]
\caption{Allocation ($\mathsterling_j^i$, $\mathcal{R}$, $\gamma^e$, $\gamma^r$)}
\begin{algorithmic}[1]
   \State  $\mathcal{A}_k \leftarrow \phi$
\State \textbf{begin}
    \State $\mathcal{U}_c^* \leftarrow Sort\_ascend$($\mathsterling_{j}^i$, $\mathcal{S}_i \cdot \chi_i^e$) \Comment Sorting based on $\chi_i^{e} \in \mathcal{P}_e$ for all $\mathcal{S}_i \in \mathsterling_{j}^i$ 
        \State $\mathcal{R}_c \leftarrow Sort\_descend$($\mathcal{R}$, $\mathcal{B}_i \cdot \chi_i^r$) \Comment Sorting based on $\chi_i^{r} \in \mathcal{P}_r$ for all $\mathcal{B}_i \in \mathcal{R}$ 
        \State $ \mathcal{I}_j \leftarrow argmax_{i} \{\chi_{i}^{r} -\chi_{i}^{e} \geq 0\}$
        \For{$k=1$ to $\mathcal{I}_j$}
        \State $\mathring{\mathcal{U}}_c \leftarrow \mathring{\mathcal{U}}_c \cup \{\mathcal{S}_k \in \mathcal{U}_c^*\}$ 
        \State $\mathring{\mathcal{R}} \leftarrow \mathring{\mathcal{R}} \cup \{\mathcal{B}_i \in \mathcal{R}_c\}$
        \State $\mathcal{A}_k \leftarrow \mathcal{A}_k \cup (\mathring{\mathcal{U}}_c, \mathring{\mathcal{R}})$
      
        \EndFor
         \State $\mathcal{U}_c^* \leftarrow \mathcal{U}_c^* \setminus \mathring{\mathcal{U}}_c$
         \State $\mathcal{R}_c \leftarrow \mathcal{R}_c \setminus \mathring{\mathcal{R}}$
        \Return ($\mathcal{A}_k$, $\mathcal{U}_{c}^*$, $\mathcal{R}_c$)

\State \textbf{end}
\end{algorithmic}
\end{algorithm}
Line 5 determines the largest index $i$ that satisfy the condition that $\chi_{i}^{r} - \chi_{i}^{e} \geq 0$. The \emph{for} loop in line 6-10 iterates over the $\mathcal{I}_j$ winning \emph{task executer}-\emph{task requester} pairs. In line 7 $\mathring{\mathcal{U}}_c$ data structure keeps track of all the winning \emph{task executers} at a particular time slot $\tau_i$ and in a particular cluster $\mathsterling_{i}^{j}$. The $\mathring{\mathcal{R}}$ data structure keeps track of all the $\mathcal{I}_j$ winning \emph{task requesters}. The $\mathcal{A}_k$ data structure in line 9 keeps track of all the winning \emph{task executer}-\emph{task requester} pairs. Line 11 and 12 removes the winning \emph{task executers} and winning \emph{task requesters} respectively from the auction market. Line 13 returns the allocation set $\mathcal{A}_k$, $\mathcal{U}_c^*$, and $\mathcal{R}_c$. 

\subsection{Analysis of STEM}
The STEM is a four stage mechanism consists of: \emph{main routine}, \emph{cluster formation}, \emph{allocation}, and \emph{payment}. So, the running time of STEM will be the sum of the running time of
\emph{main routine}, \emph{cluster formation}, \emph{allocation}, and \emph{payment}. Line 2 of the \emph{main routine} is bounded above by
O(1). The \emph{for} loop in line 3 executes for $s+1$ times as we have $s$ time slots. For the analysis purpose, WLOG we have $n \geq m$. Line 4 will take $O(n)$ time as there are \emph{n} task executers. Line 5 is bounded above by $O(m)$ as there are \emph{m} task requesters. Line 6-8 is bounded above by $O(n)$. The time taken by line 9-11 in worst case is given by $O(m)$. Considering the \emph{cluster formation} phase, in a given time horizon it is bounded above by $O(c \times k \times n) = O(ckn)$; where $c$ is number of iterations for which the change in the clusters are to be calculated. The sorting in line 14 and 15 is bounded above by $O(n\lg n)$ and $O(m \lg m)$ respectively. Now, talking about the \emph{payment phase} motivated by \cite{Jbre_INte_2005}, for \emph{k} different clusters line 3-28 will take $O(k \times \kappa \times n^2) = O(k\kappa n^2)$; where $\kappa$ is the patience bound. Line 29 in the \emph{payment phase} calls the \emph{allocation phase} that will contribute $O(n \lg n) + O(m \lg m)$ in the worst case. So, the overall time complexity is dominated by the \emph{payment phase} and is given as $O(k \times \kappa \times n^2)= O(k\kappa n^2)$. Line 17-20 in \emph{main routine} phase is bounded above by $O(n)$. Line 25 and 26 takes $O(n)$ and $O(m)$ time respectively. The overall running time of the STEM is: $O(n) + O(m) +  O(n) + O(ckn)+ O(n\lg n) + O(m \lg m) + O(k \kappa n^2) +O(n) + O(n) + O(m)$ = $O(k \kappa n^2)$. The analysis is carried out by considering the case $n \geq m$, similarly the case with $m \geq n$ can be tackled and will result in $O(k \kappa m^2)$.

\begin{lemma}\label{l1}
Agent i can't gain by misreporting their arrival time or departure time or both.
\end{lemma}	
\begin{proof}
As the agents can mis-report the arrival time or the departure time, so the proof can be illustrated into two parts considering both the cases separately.
\begin{itemize}
\item \textbf{Case 1 ($\hat{a}_{i}^{e} \neq a_i^e$)}: Fix $d_i^e$, $\tau_i$. Let us suppose an agent \emph{i} reports the arrival time as $\hat{a}_i^e$ such that $\hat{a}_{i}^{e} \neq a_i^e$ or in more formal sense $\hat{a}_{i}^{e} > a_i^e$. It means that, an agent \emph{i} will be aligned with more number of time slots before winning when reporting $\hat{a}_i^e$ than in the case when reporting \emph{truthfully} $i.e.$ $a_i^e$ as shown in \textbf{Fig.} \ref{fig:11}. Now, it is seen from the construction of the payment function that the agent \emph{i} will be paid less than or equal to the payment he/she (henceforth he) is receiving when reporting \emph{truthfully}.   
\begin{figure}[H] 
 \includegraphics[scale=0.55]{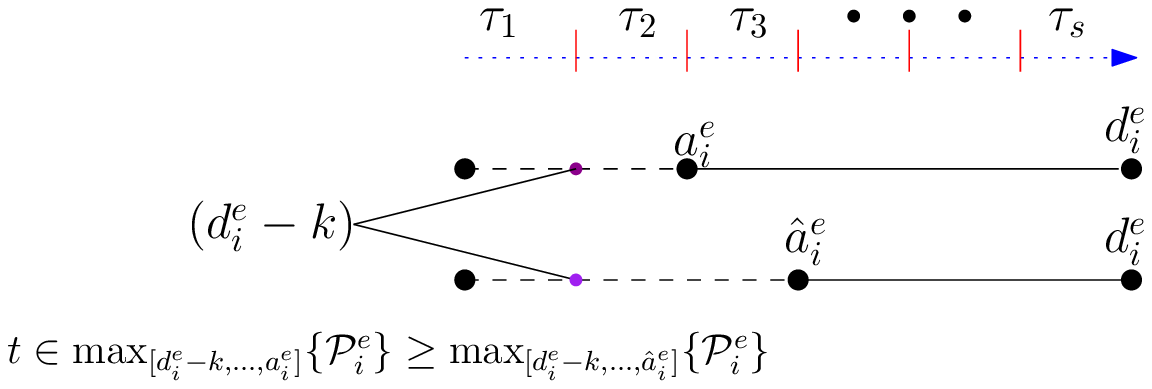}
 \caption{An agent \emph{i} mis-reporting arrival time $a_{i}^{e}$}
 \label{fig:11}

\end{figure}
\item \textbf{Case 2 ($\hat{d}_{i}^{e} \neq d_i^e$)}: Fix $a_i^e$, $\tau_i$. Let us suppose an agent \emph{i} reports the departure time as $\hat{d}_i^e$ such that $\hat{d}_{i}^{e} \neq d_i^e$ or in more formal sense $\hat{d}_{i}^{e} < d_i^e$. It means that, an agent \emph{i} will be aligned with more number of time slots before becoming inactive when reporting $\hat{d}_i^e$ than in the case when reporting truthfully $i.e.$ $d_i^e$ as shown in \textbf{Fig. \ref{fig:2}}. Now, it is seen from the construction of the payment function is that the agent \emph{i} will be paid less or equal to the payment he is paid when reporting \emph{truthfully}. 
\begin{figure}[H]
\centering
 \includegraphics[scale=0.55]{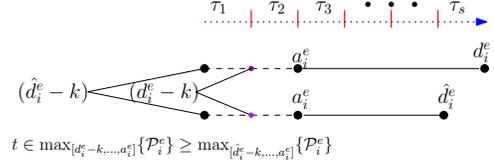}
 \caption{An agent \emph{i} mis-reporting departure time $d_{i}^{e}$}
 \label{fig:2}
\end{figure}
\end{itemize}
Considering the case 1 and case 2 above, it can be concluded that any agent \emph{i} can't gain by mis-reporting \emph{arrival time} or \emph{departure time}. The proof is carried out by considering the \emph{task executers}, similar argument can be given for the \emph{task requesters}. This completes the proof.
\end{proof}
\begin{lemma}\label{l2}
Agent i can't gain by misreporting his/her bid value.
\end{lemma}
\begin{proof} Considering the case of \emph{task executers}. Fix the time slot $\tau_i \in \varmathbb{T}$ and the cluster.\\ 
\textbf{Case 1:}\\
 Let us suppose that the $i^{th}$ winning \emph{task executer} deviates and reports a bid value $\hat{\upsilon}_{i}^{e} > \upsilon_{i}^{e}$. As the \emph{task executer} was winning with $\upsilon_{i}^{e}$, with $\hat{\upsilon}_{i}^{e}$ he would continue to win and his utility $\hat{\varphi}_i^e = \varphi_i^e$. If instead he reports
 $\hat{\upsilon}_{i}^{e} < \upsilon_{i}^e$. Again two cases can happen. He can still win. If he wins his utility, according to the definition
 will be $\hat{\varphi}_i^e = \varphi_i^e$. If he loses his utility will be $\hat{\varphi}_i^e = 0 < \varphi_i^e$.\\  
\textbf{Case 2:}\\
 If the $i^{th}$ \emph{task executer} was losing with $\upsilon_{i}$ let us see whether he would gain by deviation. If he reports 
 $\hat{\upsilon}_{i}^{e} < \upsilon_{i}^{e}$, he would still lose and his utility $\hat{\varphi}_i^e = 0 = \varphi_i^e$. If instead he reports $\hat{\upsilon}_{i}^{e} > \upsilon_{i}^e$.
 Two cases can occur. If he still loses his utility $\hat{\varphi}_i^e = 0 = \varphi_i^e$. But if he wins, then he had to beat some valuation
$\upsilon_{j}^{e} > \upsilon_{i}^e$ and hence $\hat{\upsilon}_{i}^{e} > \upsilon_{j}^{e}$. Now as he wins his utility $\hat{\varphi}_i^{e} = \mathcal{P}_{i}^{e} - \upsilon_i^e = \upsilon_j^e - 
\upsilon_i^e < 0$. So he would have got a negative utility. Hence no gain is achieved. \\
Considering the case 1 and case 2 above, it can be concluded that any agent \
\emph{i} can’t gain
by mis-reporting his bid value. The proof is carried out by considering
the \emph{task executers}, similar argument can be given for the \emph{task requesters}. This completes
the proof.
\end{proof}

\begin{lemma}\label{l3}
STEM is weakly Budget balanced.
\end{lemma}
\begin{proof}
Fix the time slot $\tau_i$ and cluster $\mathsterling_{j}^i$. This corresponds to the case when the sum of all the monetary transfers of all the agents type profiles is less than or equal to \emph{zero}. Now, the construction of our STEM is such that, any \emph{task executer} and \emph{task requester} is paired up only when $\mathcal{S}_i \cdot \mathcal{P}_i^{e} - \mathcal{B}_i \cdot \mathcal{P}_i^r \geq 0$. It means that, for any \emph{task executer}-\emph{task requester} pair there exist some surplus. In the similar fashion, in a particular time slot $\tau_i$ and in a particular cluster considering all the agents, $\sum_{i}\mathcal{S}_i \cdot \mathcal{P}_i^{e} - \sum_{i}\mathcal{B}_i \cdot \mathcal{P}_i^r \geq 0$. Hence, the sum total of payments made to the \emph{task executers} is at least as high as the sum total of the payments received by the \emph{task requesters}and their is a \emph{surplus}. Hence, it is proved that the STEM is budget balanced for a particular time slot $\tau_i$ and for a particular cluster. From our claim it must be true for any time slot $\tau_i \ldots \tau_s$ and any cluster. This completes the proof.       
\end{proof} 

\begin{lemma}\label{l4}
STEM is individual rational.
\end{lemma}
\begin{proof}
Fix the time slot $\tau_i$ and cluster $\mathsterling_{j}^i$. Individual rationality means that each agent gains a utility that is no less than he would get without participating in a mechanism. Considering the case of \emph{task requester}, when the \emph{task requester} is winning then it is ensured that he has to pay an amount $\mathcal{P}_i^r$ such that $\hat{\upsilon}_{i}^{r} \geq \mathcal{P}_i^r$. From this inequality it is clear that the winning \emph{task requester} has to pay amount less than his bid value. So, in this case it can be concluded that $\varphi_i^r = \hat{\upsilon}_{i}^{r} - \mathcal{P}_i^r \geq 0$. Moreover, if the \emph{task requester} is losing in that case his utility is 0. From our claim it must be true for any time slot $\tau_i \ldots \tau_s$ and any cluster. Similar argument can be given for the \emph{task executers}. This completes the proof.   
\end{proof}

\section{Conclusion and future works}
An incentive compatible mechanism is proposed in this paper to circumvent the location information in online double auction setting for the participatory sensing. In our future work we will focus on investigating the quality consequence in this environment. Another interesting direction is to find algorithms when the task requesters have some limited budgets.
\bibliographystyle{plain}
\bibliography{phd}
\end{document}